\documentclass[letterpaper,10pt,conference]{ieeeconf}

\IEEEoverridecommandlockouts
\usepackage{amsmath,amssymb,graphicx,booktabs,cite,microtype}
\usepackage[hidelinks]{hyperref}
\newtheorem{theorem}{Theorem}
\newtheorem{lemma}{Lemma}
\newtheorem{corollary}{Corollary}
\newtheorem{proposition}{Proposition}
\newtheorem{definition}{Definition}

\title{\LARGE \bf
Connectivity Preservation and Graph Stretching\\
in Range-Only Swarm Dispersion
}
\author{Ariel Barel$^{1}$%
\thanks{$^{1}$Faculty of Computer Science, Technion Israel Institute of Technology, Haifa, Israel. Email: arielba@technion.ac.il} }
\begin{document}
\maketitle
\thispagestyle{empty}
\pagestyle{empty}

\begin{abstract}
We study connectivity-preserving finite-jump dispersion of anonymous, identical, and oblivious agents under an idealized range-only sensing model. Each agent measures only the distances to its visible neighbors, without bearings, identifiers, communication, memory, or a shared coordinate system. We derive the largest isotropic displacement certifiable as safe from these measurements alone. The resulting rule requires only the distance to the farthest visible neighbor: each agent selects a random direction and moves by half of its remaining visibility margin. The rule preserves every existing visibility edge under synchronous finite motion and therefore preserves connectivity. For two agents, we prove positive conditional drift in squared distance, almost-sure convergence to the visibility boundary, and finite expected time to reach any fixed neighborhood of that boundary. A one-million-run Monte Carlo experiment agrees with the exact first-round moments and estimates approximately $9.5$ rounds to reach distance $0.97V$ from coincident initial positions; an independent Bellman-equation computation gives the same estimate. For general swarms, $1{,}000$ runs across five initial-topology classes reproduce the deterministic safety guarantee at implementation level and reveal a consistent topology-dependent ordering of attainable diameter under the tested protocol. These results provide a theoretical foundation for connectivity-preserving multi-robot dispersion under minimal sensing, while isolating the guarantees achievable from anonymous range measurements alone.
\end{abstract}

\section{Introduction}
Without directional information, a mobile agent cannot identify a direction away from a neighbor, reconstruct local geometry, or determine whether a visible edge is redundant for global connectivity. Under finite synchronous motion, losing one critical edge can disconnect a swarm. We study anonymous, identical, oblivious agents that know only the current distances to visible neighbors. Agents have no identifiers, memory, communication, compass, common coordinates, or bearings.

We ask for the largest finite displacement certifiable from scalar ranges alone while preserving every current edge for every feasible bearing arrangement, and whether repeated maximal safe jumps create systematic stretching. Here, \emph{dispersion} means increasing the spatial extent of the connected swarm and stretching preserved edges. It does not mean maximizing minimum pairwise distance; individual pairs may temporarily approach, and collisions are not modeled.

If the farthest visible neighbor of agent $i$ is at distance $d_i^{\max}$ and the visibility radius is $V$, our rule jumps $(V-d_i^{\max})/2$ in an independent uniform direction. This is the largest bearing-independent isotropic displacement certifiable from current ranges. It preserves every current edge and hence connectivity. An edge at distance $V$ permanently immobilizes both endpoints, and every absorbing configuration contains at least $n/2$ boundary edges.

Preserving every edge is stronger than preserving connectivity alone, but anonymous current ranges do not reveal which edges are globally redundant. Thus preserving a strict subset cannot be locally certified under the stated model. For two agents, squared distance has positive conditional drift below $V$, yielding almost-sure convergence and finite expected near-boundary hitting time. For larger swarms we provide a multi-run empirical topology study rather than a convergence claim.

Our main contributions are as follows.
\begin{itemize}
    \item We derive the maximal isotropic displacement that can be certified using anonymous current range measurements. Although an agent can observe several ranges, the safe-action certificate depends only on the farthest one.
    \item We prove deterministic preservation of every current visibility edge for arbitrary swarm size and every realization of the random directions. Connectivity preservation follows immediately for an initially connected swarm.
    \item We prove exact locking at the visibility boundary and show that every absorbing configuration contains at least $n/2$ boundary edges.
    \item For two agents, we prove positive conditional drift in squared distance, almost-sure convergence to the visibility boundary, and finite expected time to reach every fixed boundary neighborhood.
    \item We complement the analytical results with large-scale Monte Carlo validation, an independent Bellman computation of the near-boundary hitting time, and extensive empirical studies showing topology-dependent stretching across multiple swarm sizes and initial graph classes.
\end{itemize}

The paper separates three levels of claim. The safety and structural results hold for arbitrary $n$. The stochastic convergence results are proved only for the analytically tractable two-agent system. For general coupled swarms, we make no convergence claim; the contribution is an empirical characterization of topology-dependent expansion and edge stretching.

\section{Related Work}
Connectivity-preserving coordination commonly uses relative positions, interaction graphs, potential fields, or learned navigation policies. Ji and Egerstedt preserve connectedness of dynamic interaction graphs \cite{ji2007distributed}. Zavlanos and Pappas use differentiable constraints and artificial potential fields \cite{zavlanos2007potential}; Dimarogonas and Johansson propose bounded distributed connectivity control \cite{dimarogonas2008decentralized}; and Li et al. learn a decentralized navigation policy \cite{li2022decentralized}. These approaches require relative positions, graph information, navigation functions, or a trained policy with richer observations.

Ando et al. restrict pair destinations to a midpoint-centered disk \cite{ando1999distributed}, but locating the midpoint requires relative direction. Gordon et al. study bearing-only and crude near/far sensing \cite{gordon2004gathering,gordon2008crude}; continuous-time bearing-only gathering appears in \cite{bellaiche2017continuous}. Flocchini et al. consider anonymous oblivious robots with limited visibility under asynchrony but assume orientation \cite{flocchini2005gathering}. The broader geometric-consensus taxonomy surveyed in \cite{barel2019cometogether} considers the corresponding bearing-only setting but does not address range-only gathering, for which no solution under this severe sensing restriction was identified at the time.

Recent work moves closer to scalar-distance sensing. Lee et al. use an ultrasonic tracker that supplies both range and bearing \cite{lee2023practical}. Brandst\"atter et al. use imperfect scalar distances, but each agent updates a directional-response model and uses state estimation and a target distance \cite{brandstatter2024flock}. Chen et al. provide range-only safety-critical formation control, but reconstruct bearing relative to a landmark and use an estimator, control barrier functions, and a prescribed formation \cite{chen2025rangeonly}. Yang et al. combine online communication learning with Gaussian processes, barrier functions, and bi-level optimization \cite{yang2024online}. Manor et al. release redundant interactions using local geometry \cite{manor2019local}, which is unavailable here.

The contribution is not the triangle inequality alone, but the maximality of the resulting action certificate under anonymous range-only information, its reduction to one scalar statistic, and the stochastic behavior induced by repeatedly exhausting that certificate.

\section{Model and Motion Rule}
Let $\mathcal A=\{1,\ldots,n\}$, $n\ge2$, and $p_i(t)\in\mathbb R^2$. The visibility graph $G(t)=(\mathcal A,E(t))$ satisfies
\begin{equation}(i,j)\in E(t)\iff\|p_i(t)-p_j(t)\|\le V.\end{equation}
Define
\begin{equation}N_i(t)=\{j\in\mathcal A\setminus\{i\}:(i,j)\in E(t)\}.\end{equation}
We assume $G(0)$ is connected. Each agent observes only the multiset of current ranges and need not associate measurements with persistent identities; it only computes the maximum.

For a pair at distance $d$, the total margin is $V-d$. In the worst case both agents jump directly apart, so each receives $(V-d)/2$. Any larger common isotropic radius permits loss of the edge.

\begin{proposition}[Maximal isotropic range-only radius]\label{prop:maximal-radius}
For a visible pair at distance $d$, the largest common isotropic displacement radius that guarantees preservation of the edge for every feasible pair of bearings and synchronous moves is
\begin{equation}
\frac{V-d}{2}.
\end{equation}
\end{proposition}
\begin{proof}
Sufficiency follows from the triangle inequality. For maximality, let both agents use any radius $r>(V-d)/2$ and consider the feasible motion in which they move directly away from one another. Their new distance is $d+2r>V$, so the edge is lost.
\end{proof}

With several neighbors, the farthest is most restrictive:
\begin{equation}
\begin{aligned}
d_i^{\max}(t)
&=\max_{j\in N_i(t)}\|p_i(t)-p_j(t)\|,\\
r_i(t)
&=\frac{V-d_i^{\max}(t)}{2}.
\end{aligned}
\label{eq:allowable-radius}
\end{equation}
\begin{corollary}[Agent allowable region]\label{cor:agent-allowable-region}
The complete isotropic range-only allowable region of agent $i$ is
\begin{equation}
AR_i(t)=D\!\left(p_i(t),\frac{V-d_i^{\max}(t)}{2}\right),
\end{equation}
where $D(p,r)$ denotes the closed disk of radius $r$ centered at $p$. Hence the full multiset of current neighbor ranges reduces to the single statistic $d_i^{\max}(t)$.
\end{corollary}
\begin{proof}
Each visible neighbor $j$ induces a concentric safe disk of radius $(V-d_{ij}(t))/2$. Their intersection is the disk with the smallest radius. Since $d\mapsto(V-d)/2$ is decreasing, this radius is induced by the farthest current neighbor.
\end{proof}

Agent $i$ samples $U_i(t)$ uniformly on the unit circle and updates
\begin{equation}p_i(t+1)=p_i(t)+r_i(t)U_i(t).\label{eq:update}\end{equation}
The boundary is selected to exhaust the certified displacement budget. This does not imply that every pairwise distance rises in every round. Safety is provided by the radius, while the random direction explores the safe set without a bearing estimate. Since all pairwise safety disks are concentric at the agent, the complete current neighborhood reduces to the single scalar $d_i^{\max}$.

Let $\mathcal F_t$ denote the history immediately before round-$t$ directions are sampled. Thus, conditional on $\mathcal F_t$, the current positions, $X_t$, $D_t$, and all step radii are fixed, while the new directions remain random. Appendix~A reviews the conditional-expectation notation used below, and Appendix~B gives the full derivation of the two-agent drift identities.

\subsection{Information Use and Design Rationale}
The rule is deliberately minimal. It does not require neighborhood ordering, persistent measurement association, geometric reconstruction, or a target formation. All pairwise range-only safety disks are centered at the current agent, so their intersection is determined by the smallest radius, equivalently by the farthest measured neighbor range. Thus the full current range multiset reduces to one scalar statistic, $d_i^{\max}(t)$.

Safety and dispersion are separated. The radius certifies that every possible direction is safe, while the random direction explores the safe set without reconstructing a bearing. Choosing the boundary uses the entire certified displacement budget. A random move can increase some neighbor distances and decrease others, but no current edge can be lost. The following analytical ablation makes the boundary choice precise.

\begin{proposition}[Step-length ablation]
For $r_i^{(\alpha)}=\alpha(V-d_i^{\max})/2$, $\alpha\in[0,1]$, every $\alpha$ preserves all edges. For two agents,
\begin{equation}\mathbb E[D_{t+1}^2-D_t^2\mid\mathcal F_t]=\frac{\alpha^2(V-D_t)^2}{2}.\end{equation}
Thus $\alpha=1$ maximizes one-step expected squared-distance progress in this safe isotropic family. The complete expansion of the relative-position update, including the vanishing linear term and the factor $\alpha^2$ in the quadratic term, is given in Appendix~B.
\end{proposition}
\begin{proof}Safety follows from $r_i^{(\alpha)}\le r_i$. Scaling both movement vectors by $\alpha$ multiplies the quadratic drift term by $\alpha^2$.\end{proof}

\section{Safety and Structural Guarantees}
\begin{theorem}[Edge preservation]\label{thm:edge}For every $t\ge0$, $E(t)\subseteq E(t+1)$.\end{theorem}
\begin{proof}For $(i,j)\in E(t)$, write $d_{ij}=\|p_i-p_j\|$. Then $r_i,r_j\le(V-d_{ij})/2$, and
\begin{equation}\|p_i(t+1)-p_j(t+1)\|\le d_{ij}+r_i+r_j\le V.\end{equation}
\end{proof}
\begin{corollary}If $G(0)$ is connected, then $G(t)$ is connected for all $t$.\end{corollary}
\begin{proof}Theorem~\ref{thm:edge} gives $E(0)\subseteq E(t)$, so $G(t)$ contains connected $G(0)$ as a spanning subgraph.\end{proof}

The edge-preservation proof is independent of the distribution of the selected directions. Uniform random directions are used for the stretching analysis, not for safety. In particular, any direction-selection policy remains safe as long as each displacement respects the same certified radius. This modularity allows the range-only safety layer to be separated from future task-specific direction policies.

\begin{theorem}[Boundary locking]\label{thm:boundary-locking}
If $\|p_i(t)-p_j(t)\|=V$ for some $(i,j)\in E(t)$, then
\begin{equation}
p_i(s)=p_i(t),\qquad p_j(s)=p_j(t)
\end{equation}
for every $s\ge t$.
\end{theorem}
\begin{proof}
Since $(i,j)\in E(t)$ has length $V$, both endpoints have a visible neighbor at distance $V$. No visible-neighbor distance can exceed $V$, so
\begin{equation}
d_i^{\max}(t)=d_j^{\max}(t)=V.
\end{equation}
Equation~\eqref{eq:allowable-radius} therefore gives $r_i(t)=r_j(t)=0$. By the update rule~\eqref{eq:update}, neither endpoint moves at round $t$. Their mutual distance consequently remains $V$, and the same argument applies recursively at every subsequent round.
\end{proof}

\begin{definition}A configuration is \emph{absorbing} if $r_i=0$ for every agent.\end{definition}
\begin{theorem}Every absorbing configuration contains at least $n/2$ edges of length $V$.\end{theorem}
\begin{proof}
Let
\begin{equation}
E_V=\{(i,j)\in E:\|p_i-p_j\|=V\}
\end{equation}
denote the set of visibility-boundary edges. In the graph $(\mathcal A,E_V)$, every vertex has degree at least one. Hence
\begin{equation}
2|E_V|=\sum_{i=1}^{n}\deg_{E_V}(i)\ge n.
\end{equation}
\end{proof}
This is conditional on absorption and does not prove general-swarm convergence. It nevertheless shows that if the swarm fully immobilizes, stretching is structural rather than transient: every agent is incident to at least one visibility-boundary edge. Exact equality is a property of the idealized model; numerical implementations must treat a tolerance band separately.

Boundary locking is irreversible in the ideal exact-arithmetic model. Once one incident edge reaches $V$, the corresponding agent has farthest-neighbor distance $V$ and can no longer move. The lower bound on boundary edges is conditional on the system reaching an absorbing configuration, but it shows that complete immobilization necessarily contains a linear amount of exact edge stretching rather than merely a transient increase in diameter. In numerical implementations, near-boundary edges must be distinguished from exact locked edges because tolerance and finite precision can otherwise create apparent locking or unlocking events.

\section{Two-Agent Stochastic Stretching}
Let $X_t=p_2(t)-p_1(t)$ and $D_t=\|X_t\|$. Then
\begin{equation}X_{t+1}=X_t+\frac{V-D_t}{2}(U_2(t)-U_1(t)).\label{eq:relative}\end{equation}
Because $D_t\le V$ by edge preservation, the process remains in the compact interval $[0,V]$. The boundary $V$ is absorbing, while the next distance below $V$ depends only on the current relative displacement and two fresh independent directions. Rotational symmetry therefore reduces the stochastic state relevant to distance evolution to the scalar $D_t$.

It is important to distinguish the mean relative-displacement vector from the mean distance. Symmetry may give a zero vector after averaging directions, but opposite relative-displacement vectors cancel whereas their nonnegative lengths do not. In general,
\[
\mathbb E[\|X_t\|]\neq \|\mathbb E[X_t]\|.
\]
The analysis below therefore works directly with $D_t=\|X_t\|$, and especially with $D_t^2$, whose one-step expansion separates into a linear term and a quadratic term.

\begin{lemma}[Positive drift]\label{lem:drift}
\begin{equation}\mathbb E[D_{t+1}^2\mid\mathcal F_t]=D_t^2+\frac{(V-D_t)^2}{2}.\label{eq:drift}\end{equation}
\end{lemma}
\begin{proof}
Conditioning on $\mathcal F_t$ and expanding the squared norm gives
\begin{align}
D_{t+1}^2
&=D_t^2+(V-D_t)\langle X_t,U_2(t)-U_1(t)\rangle\nonumber\\
&\quad+\frac{(V-D_t)^2}{4}\|U_2(t)-U_1(t)\|^2.
\end{align}
Uniformity gives $\mathbb E[U_i(t)]=0$, so the conditional expectation of the cross term vanishes. Independence gives $\mathbb E[\langle U_1(t),U_2(t)\rangle]=0$, and therefore
\begin{equation}
\mathbb E[\|U_2(t)-U_1(t)\|^2]=2,
\end{equation}
which yields the claim.
\end{proof}

\begin{theorem}[Almost-sure convergence]For every $D_0\in[0,V]$, $D_t\to V$ almost surely.\end{theorem}
\begin{proof}Summing \eqref{eq:drift} gives
\begin{equation}\sum_{t=0}^{\infty}\mathbb E[(V-D_t)^2]\le2(V^2-D_0^2)<\infty.\end{equation}
To turn this summability bound into an almost-sure convergence statement, we first use Markov's inequality to control the probability of a non-negligible gap from \(V\), and then apply the Borel--Cantelli lemma to show that such gaps occur only finitely often almost surely.

Markov's inequality converts the second-moment bound into $\sum_t\Pr(V-D_t\ge\varepsilon)<\infty$. The first Borel--Cantelli lemma \cite{durrett2019probability} then implies that, for each fixed $\varepsilon>0$, the event $V-D_t\ge\varepsilon$ occurs only finitely often almost surely. Applying this statement to the countable sequence $\varepsilon=1/m$, $m\in\mathbb N$, proves $D_t\to V$ almost surely. Appendix~C expands every step, including the tower property, the telescoping sum, Markov's inequality, and the Borel--Cantelli argument.\end{proof}
\begin{proposition}[No finite exact hit]\label{prop:no-finite-exact-hit}
If $D_0<V$ and the directions are continuously uniform, then
\begin{equation}
\Pr(\tau_V<\infty)=0,
\qquad
\tau_V=\inf\{t\ge0:D_t=V\}.
\end{equation}
\end{proposition}
\begin{proof}
Condition on $\mathcal F_t$ and $D_t<V$. Equality at round $t+1$ requires maximum relative-step magnitude and equality in the triangle inequality. For $X_t\neq0$, the directions must be antipodal and aligned with the current separation line; for $X_t=0$, antipodality is still required. These direction pairs form a measure-zero subset of $\mathbb S^1\times\mathbb S^1$, so the conditional probability of an exact hit at the next round is zero whenever $D_t<V$. Since the set of finite rounds is countable, the probability of an exact hit at any finite round is also zero. Appendix~D gives the measure-zero and countable-union argument in full.
\end{proof}
Hence the agents almost surely approach the boundary without hitting it exactly in finite time, and their allowable radii converge to zero.

Let $\tau_\varepsilon=\inf\{t:D_t\ge V-\varepsilon\}$.
\begin{theorem}[Near-boundary hitting time]\label{thm:near-boundary-hitting}
If $D_0=d_0<V-\varepsilon$, then
\begin{equation}\mathbb E[\tau_\varepsilon\mid D_0=d_0]\le\frac{2(V^2-d_0^2)}{\varepsilon^2}.\end{equation}
\end{theorem}
\begin{proof}
Let $\sigma_T=T\wedge\tau_\varepsilon=\min\{T,\tau_\varepsilon\}$. This bounded stopping time records whichever occurs first: entry into the target region or the deterministic cutoff $T$. On $\{t<\tau_\varepsilon\}$, the agents still execute the original maximal safe jump of radius $(V-D_t)/2$; the proof does not replace this radius by $\varepsilon/2$. Since $V-D_t>\varepsilon$ before target entry, Lemma~\ref{lem:drift} gives conditional drift at least $\varepsilon^2/2$. Summing the stopped drift yields
\begin{equation}
\mathbb E[D_{\sigma_T}^2]-d_0^2
\ge \frac{\varepsilon^2}{2}\mathbb E[T\wedge\tau_\varepsilon].
\end{equation}
Because $D_{\sigma_T}\le V$,
\begin{equation}
\mathbb E[T\wedge\tau_\varepsilon]
\le \frac{2(V^2-d_0^2)}{\varepsilon^2}.
\end{equation}
Letting $T\to\infty$ and applying monotone convergence proves the result. A line-by-line stopped-drift derivation and a short explanation of monotone convergence appear in Appendix~D.
\end{proof}
The bound establishes finite expected hitting time but can be highly conservative because it discards the substantially larger drift obtained when $D_t$ is far below $V-\varepsilon$.
For $D_0=0$, both agents initially move with radius $V/2$. If $\Delta$ is the difference between their independent uniform directions, then
\begin{equation}
D_1=V\left|\sin\left(\frac{\Delta}{2}\right)\right|.
\end{equation}
The two endpoints form a random chord of the radius-$V/2$ circle. Let $\Phi\in[0,\pi]$ be the smaller relative angle between the two sampled directions. Then $\Phi$ is uniform on $[0,\pi]$ and $\mathbb E[\Phi]=\pi/2$, but the normalized chord length is
\[
\frac{D_1}{V}=\sin\!\left(\frac{\Phi}{2}\right).
\]
Consequently, the mean distance is not the chord associated with the mean angle: $\mathbb E[D_1/V]\neq\sin(\mathbb E[\Phi]/2)$. The distinction is due to the nonlinear, concave dependence of chord length on $\Phi$. Appendix~E derives the chord formula and distinguishes the mean chord length, its second moment, and its root-mean-square length. Direct integration gives
\begin{equation}\mathbb E[D_1]=2V/\pi,\qquad\mathbb E[D_1^2]=V^2/2.\label{eq:moments}\end{equation}

\section{Experimental Evaluation}
The general-swarm study uses five $N=10$ topology groups: complete, path, and sparse, medium, and dense random-connected graphs. Each has $200$ independent runs of $5{,}000$ rounds, totaling five million rounds. We record connectivity, edge changes, normalized diameter, and stretching of a fixed initial spanning tree. Diameter sensitivity uses $100$ runs per group for $N=20$ and $N=50$.

We use $V=50$. Complete configurations lie in a radius-$10$ disk; paths use consecutive edges of length $0.9V$. Random-connected candidates use sampling-disk radii $90$, $55$, and $42$, are retained only if connected, and are classified by initial edge density as sparse $[0.05,0.35)$, medium $[0.35,0.65)$, or dense $[0.65,0.95)$. The base seed is $72{,}010{,}001$. For general swarms, this is an empirical topology study rather than a convergence claim.

The topology classes serve distinct purposes. Complete and path graphs are structural extremes: preserving every pairwise edge confines a complete graph to diameter $V$, whereas a path constrains only local consecutive separations. The three random-connected strata test whether the resulting geometry changes gradually with initial edge density rather than only between two hand-selected cases. The $N=20$ and $N=50$ experiments test whether the qualitative ordering persists beyond the main $N=10$ setting.

For every run, a spanning tree is selected once from the initial visibility graph. Its $N-1$ endpoint pairs are then frozen as measurement identities: the tree is not recomputed and does not affect agent motion. Only the lengths of those same initial edges are tracked throughout the execution. This provides an equal-size, connectivity-sufficient set of local observables for every topology group. Appendix~F gives the construction and interpretation in detail. Aggregate statistics include medians, interquartile ranges, and $5$th--$95$th percentile ranges. Time-series measurements are sampled every ten rounds. The same base seed and fixed group-specific offsets are used throughout.

No edge or connectivity loss was observed, reproducing the deterministic guarantee at implementation level over five million rounds. This is an implementation check rather than independent evidence for the theorem. Beyond safety, the experiments quantify the geometric cost of preserving the initial constraints. Complete graphs remain confined to diameter $V$, whereas paths expand over many visibility ranges. For $N=10$, median final diameters of dense, medium, and sparse random-connected groups are approximately $1.65V$, $2.2V$, and $3.6V$. Diameter measures global expansion, while stretching of a fixed initial spanning tree measures deformation of an initially connectivity-sufficient structure.

\begin{figure}[t]\centering
\includegraphics[width=0.98\columnwidth]{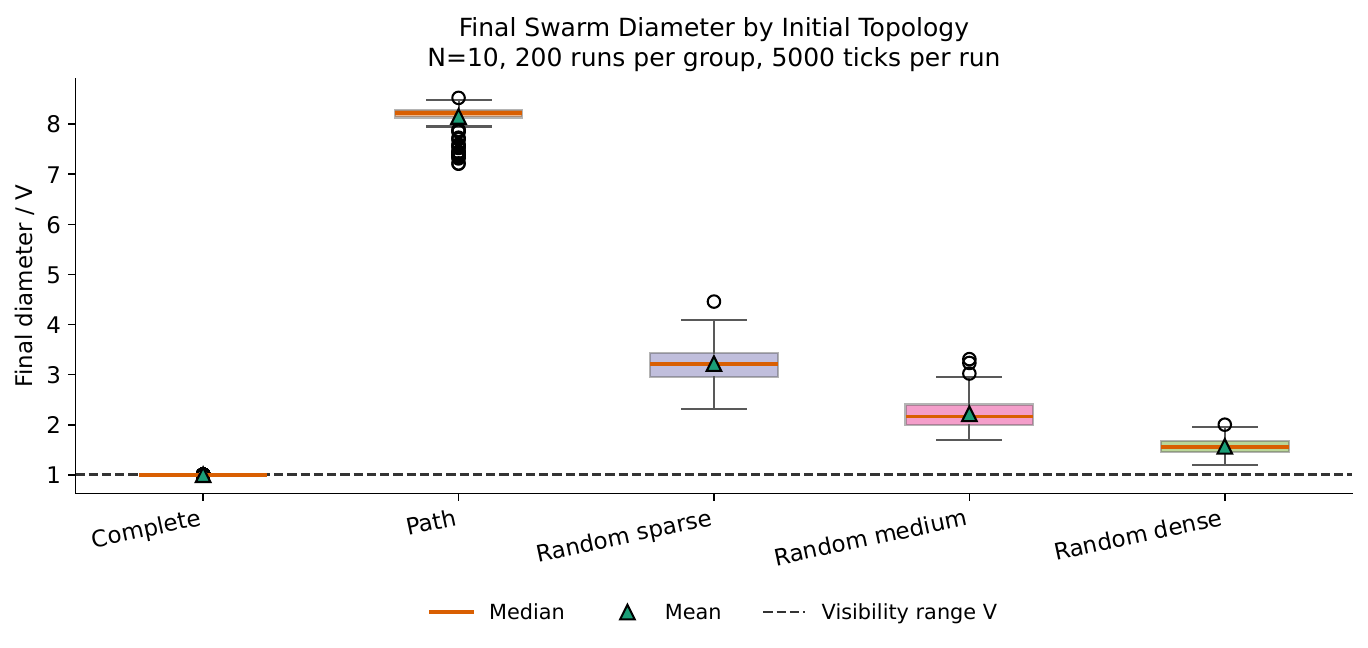}
\caption{Final normalized diameter for five $N=10$ topology groups, $200$ runs per group.}\label{fig:diameter}\end{figure}

\begin{figure}[t]\centering
\includegraphics[width=0.99\columnwidth]{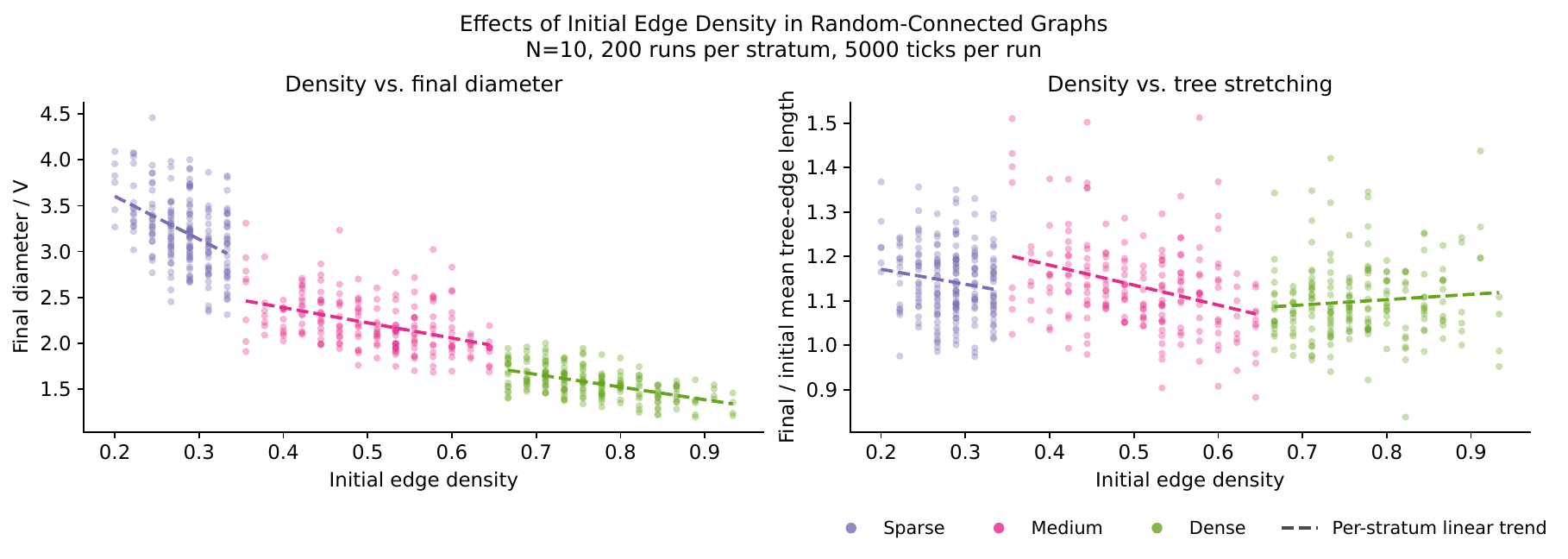}
\caption{Initial edge density versus final diameter and fixed-tree stretching for random-connected $N=10$ swarms.}
\label{fig:density}\end{figure}

Within the random-connected family, greater initial density is associated with smaller final diameter under the tested initialization protocol. Fixed-tree stretching is more variable, showing that global spatial extent and local edge stretching are related but distinct observables.

Figure~\ref{fig:density} compares two outcomes within the random-connected family. Final diameter measures global spatial expansion, while the final-to-initial mean length ratio of the fixed tree measures stretching of a connectivity-sufficient initial structure. Their visibly different variability shows that global expansion and local edge stretching are related but not interchangeable summaries.

Most observed expansion occurs during the first $500$ rounds, followed by substantially slower variation through round $5{,}000$. The percentile bands computed in the full evaluation show that the topology-dependent differences persist across independent runs rather than being driven by isolated executions. Fixed-tree edges stretch substantially in every topology group, with the strongest normalized growth for path and sparse random-connected initial graphs. These trends support the interpretation that global diameter and deformation of the preserved connectivity structure capture distinct aspects of the dynamics.

Across all tested sizes, median final diameter follows
\begin{equation}\text{Complete}<\text{Dense}<\text{Medium}<\text{Sparse}<\text{Path}.\end{equation}
For paths, the median rises from approximately $17.22V$ at $N=20$ to $44.23V$ at $N=50$.

\begin{figure}[t]\centering
\includegraphics[width=0.98\columnwidth]{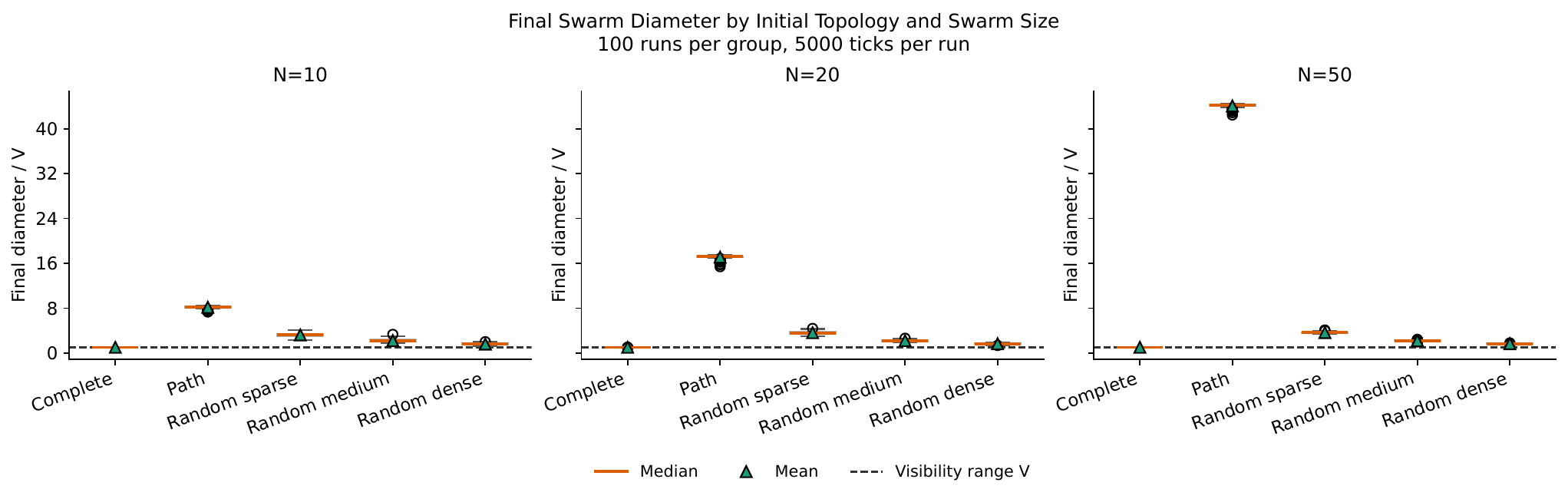}
\caption{Final normalized diameter for $N=10$, $N=20$, and $N=50$. The $N=10$ panel uses $200$ runs per group; the $N=20$ and $N=50$ panels use $100$ runs per group.}\label{fig:size}\end{figure}

Appendix~G provides representative geometric executions that complement these aggregate statistics. The illustrations contrast a complete initial graph, whose preserved pairwise constraints confine the swarm to diameter at most $V$, with a sparse path-like initial graph that expands over several visibility ranges. These examples are included only for geometric interpretation; the multi-run results above remain the empirical basis for the topology comparisons.

For two agents we simulate one million runs with $D_0=0$, $V=1$, and $100$ rounds. The empirical first-round moments are $0.636830V$ and $0.500073V^2$, matching \eqref{eq:moments} and directly validating the simulation. For $D_t\ge0.97V$, first-hitting time has mean $9.5143$, median $7$, standard deviation $8.6483$, and quantiles $(1,3,13,26)$ at probabilities $(0.05,0.25,0.75,0.95)$. Only $1.2\times10^{-5}$ of runs have not entered by round $100$. The mean distance first exceeds $0.97V$ at round $21$; at that round $82.6507\%$ of runs are currently in the region and $90.8140\%$ have entered at least once. The quantities differ because $D_t$ is not monotone along individual trajectories.

\begin{figure*}[t]
\centering
\begin{minipage}[t]{0.485\textwidth}
\centering
\includegraphics[width=\linewidth]{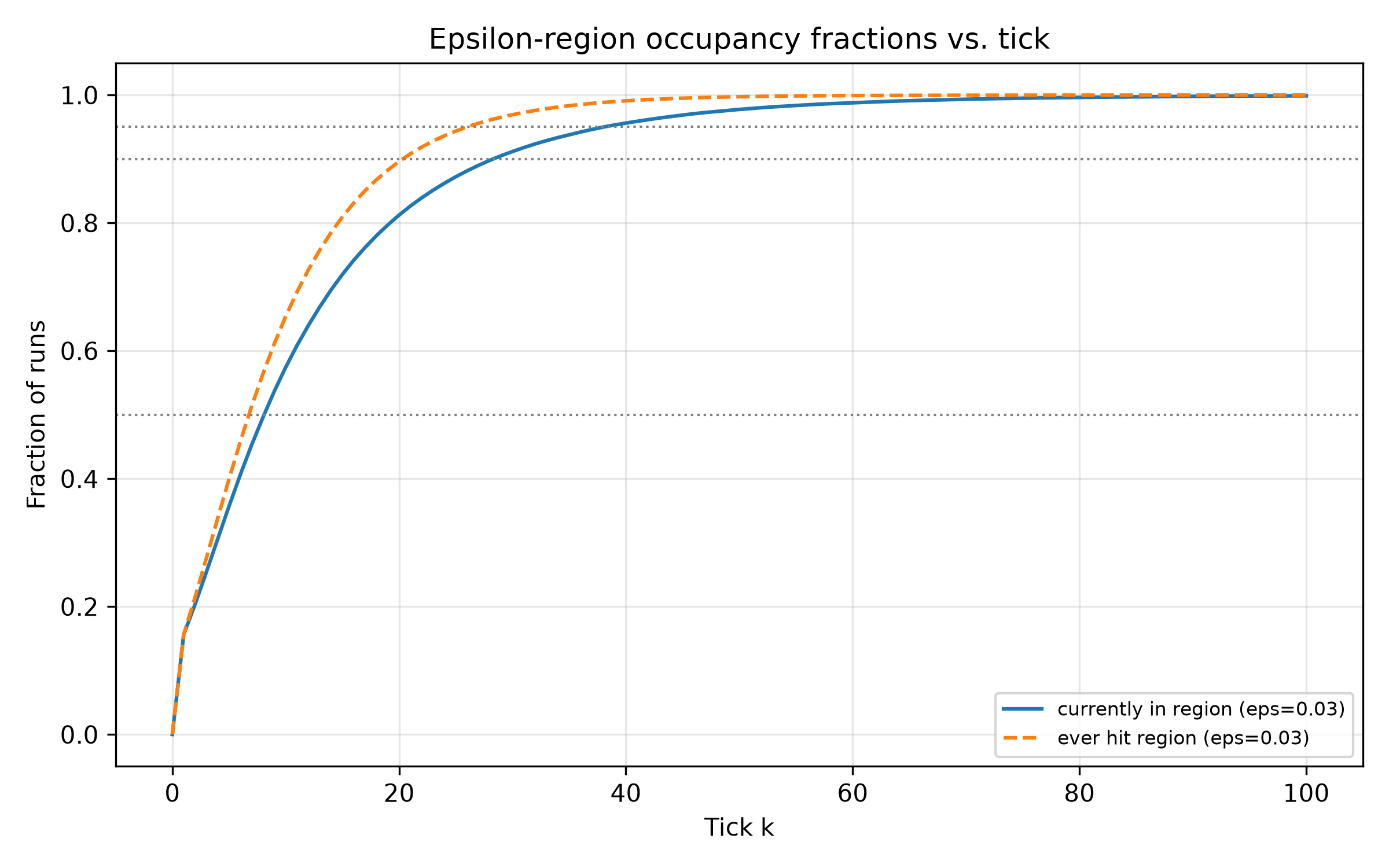}
\vspace{-1mm}

{\footnotesize (a) Current occupancy versus cumulative entry.}
\end{minipage}
\hfill
\begin{minipage}[t]{0.485\textwidth}
\centering
\includegraphics[width=\linewidth]{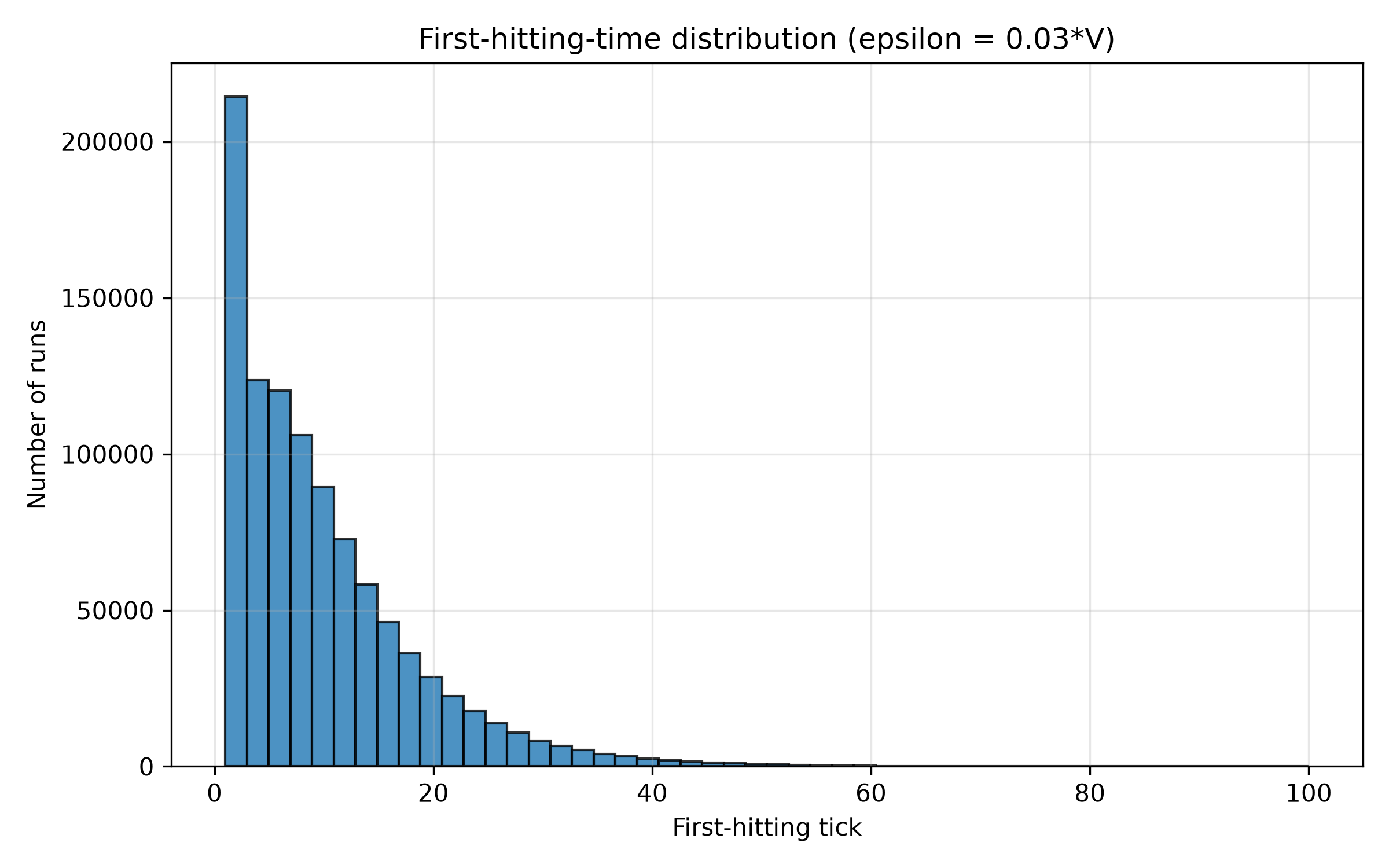}
\vspace{-1mm}

{\footnotesize (b) First-hitting-time distribution.}
\end{minipage}
\caption{Finite-round behavior in one million two-agent runs for $D_t\ge0.97V$.}
\label{fig:twoagent}
\end{figure*}

The empirical first-round moments closely reproduce the exact values and therefore provide a direct implementation check. The hitting-time distribution is concentrated in early rounds but has a right tail. Current occupancy and cumulative first entry are distinct because an individual trajectory can leave the $0.97V$ region after first entering it. The cumulative entry fraction is monotone by definition, whereas current occupancy need not be monotone.

The theoretical bound is $2/(0.03)^2\approx2222.22$. To retain the exact state-dependent transition law, normalize $V=1$ and write $a(z)=(1-z)/2$. By translation and rotational symmetry, choose coordinates in which the current agent positions are
\[
p_1=(0,0)^\top,\qquad p_2=(z,0)^\top.
\]
This coordinate choice preserves their normalized separation $z$. For independent $\theta_1,\theta_2\sim\operatorname{Unif}[0,2\pi)$, write $u(\theta)=(\cos\theta,\sin\theta)^\top$. The post-move positions are
\[
p_1'=a(z)u(\theta_1),\qquad
p_2'=(z,0)^\top+a(z)u(\theta_2).
\]
Hence the new relative-displacement vector is
\[
p_2'-p_1'=
\begin{pmatrix}
z+a(z)(\cos\theta_2-\cos\theta_1)\\
a(z)(\sin\theta_2-\sin\theta_1)
\end{pmatrix}.
\]
Taking its Euclidean norm gives the next normalized distance
\begin{equation}
\begin{aligned}
F(z,\theta_1,\theta_2)
&=\Bigl[
\bigl(z+a(z)(\cos\theta_2-\cos\theta_1)\bigr)^2
\\
&\qquad
+a(z)^2\bigl(\sin\theta_2-\sin\theta_1\bigr)^2
\Bigr]^{1/2}.
\end{aligned}
\label{eq:normalized-transition}
\end{equation}
For $Z_t=D_t/V$ and a normalized target $b=0.97$, define \[ \tau_b^{Z}=\inf\{t\ge0:Z_t\ge b\}, \qquad h_b(z)=\mathbb E[\tau_b^{Z}\mid Z_0=z], \] with $h_b(z)=0$ for $z\ge b$. Here, $\tau_{0.97}^{Z}=\tau_{0.03V}$ under the earlier boundary-margin notation.
Conditioning on the first transition gives, for $z<b$,
\begin{equation}
h_b(z)=1+\frac{1}{(2\pi)^2}
\int_0^{2\pi}\!\int_0^{2\pi}
 h_b(F(z,\theta_1,\theta_2))\,d\theta_1d\theta_2,
\end{equation}
Equation~(26) has a direct first-step interpretation. Starting from a normalized distance $z<b$, the agents first execute one round, which contributes the term $1$. Their independently sampled directions $\theta_1$ and $\theta_2$ determine the next normalized distance $F(z,\theta_1,\theta_2)$ through~\eqref{eq:normalized-transition}. From that new state, the expected number of additional rounds required to reach the target is $h_b(F(z,\theta_1,\theta_2))$. Since both directions are uniformly distributed, the double integral averages this remaining time over all possible direction pairs, with the factor $1/(2\pi)^2$ providing the required normalization. If $F(z,\theta_1,\theta_2)\ge b$, the target is reached during the current round, so the remaining-time contribution for that direction pair is zero. Thus, the equation states that the expected hitting time equals one current round plus the expected remaining hitting time from the random next state.

We next convert the continuous first-step equation into a finite linear system. Partition the transient interval $[0,b)$ into $M$ cells
\[
C_m=\left[\frac{(m-1)b}{M},\frac{mb}{M}\right),
\qquad m=1,\ldots,M,
\]
and let $z_m\in C_m$ be the representative state used for cell $C_m$. For each angular variable, use the periodic quadrature nodes
\[
\theta_q=\frac{2\pi q}{Q},
\qquad q=0,\ldots,Q-1.
\]
The endpoint $2\pi$ is not repeated because it represents the same direction as $0$. For every representative state $z_m$ and every direction pair $(\theta_q,\theta_r)$, compute
\[
z'=F(z_m,\theta_q,\theta_r).
\]
If $z'\ge b$, the target is reached in the current round and the remaining-time contribution is zero. If $z'<b$, assign $z'$ to its destination cell $C_j$. The resulting transition weight is
\[
P_{mj}^{(M,Q)}
=\frac{1}{Q^2}
\#\left\{(q,r):F(z_m,\theta_q,\theta_r)\in C_j\right\}.
\]
Thus, $P_{mj}^{(M,Q)}$ approximates the probability of moving from transient cell $m$ to transient cell $j$ in one round. The matrix is substochastic because transitions with $z'\ge b$ leave the transient state space; specifically,
\[
1-\sum_{j=1}^{M}P_{mj}^{(M,Q)}
\]
is the quadrature approximation of the one-round target-entry probability from cell $m$. If $h_m$ denotes the approximate expected remaining hitting time associated with cell $C_m$, the discrete first-step equations are
\[
h_m=1+\sum_{j=1}^{M}P_{mj}^{(M,Q)}h_j,
\qquad m=1,\ldots,M.
\]
Collecting these equations gives $\mathbf h=\mathbf 1+P^{(M,Q)}\mathbf h$, or equivalently
\begin{equation}
(I-P^{(M,Q)})\mathbf h=\mathbf 1.
\end{equation}
Solving this system gives approximately $9.50$ rounds, independently corroborating the Monte Carlo estimate for the same first-hitting-time quantity.

\begin{table}[t]\centering
\caption{Bellman approximation of $h_{0.97}(0)$.}\label{tab:bellman}
\begin{tabular}{rrc}\toprule $M$&$Q$&$h_{0.97}^{(M,Q)}(0)$\\\midrule
400&256&9.4896\\800&512&9.4943\\1200&1024&9.4964\\\bottomrule\end{tabular}\end{table}

The Bellman values stabilize near $9.50$ under successive state and angular refinements. The calculation retains the full state-dependent one-step transition law, whereas the analytical hitting-time bound replaces the drift by its minimum value before target entry. Their large numerical difference is therefore expected: the theorem proves finite expected hitting time but is not intended as a tight finite-round prediction. The agreement between Bellman and Monte Carlo concerns the same first-hitting-time quantity and provides two independent numerical estimates. For the finest discretization, the maximum absolute residual of the computed linear system is below $5\times10^{-14}$. This verifies the numerical solve but does not bound discretization error. A separate one-million-run implementation yields mean $9.5033$ with estimated standard error $0.0086$ and reproduces the reported median and quantiles.

\begin{figure}[t]\centering
\includegraphics[width=0.98\columnwidth]{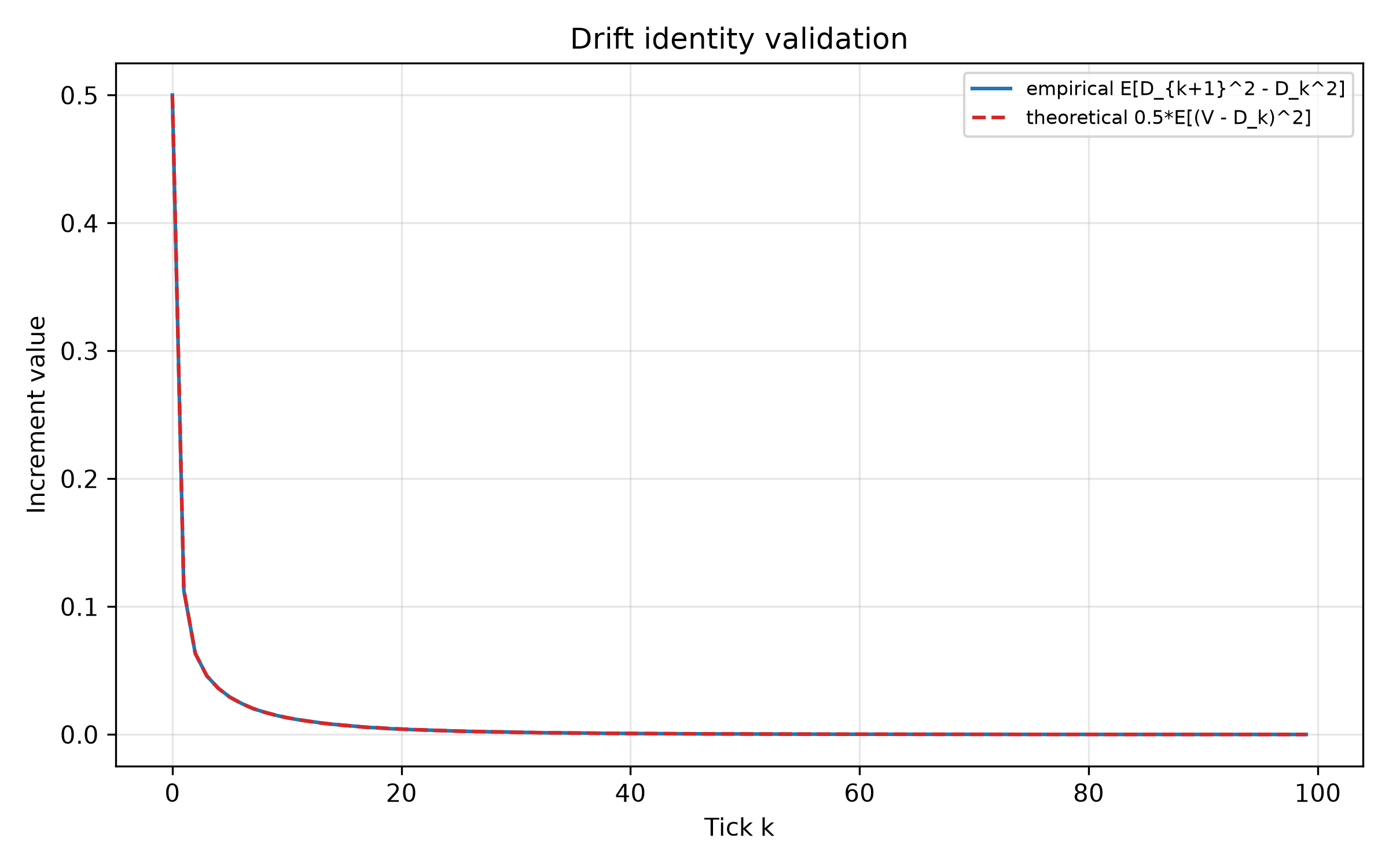}
\caption{Monte Carlo validation of Lemma~\ref{lem:drift}.}\label{fig:drift}\end{figure}

\begin{table}[t]
\centering
\caption{Theory and numerical results for two initially coincident agents.}
\label{tab:two-agent-comparison}
\begin{tabular}{lcc}
\toprule
Quantity & Analysis & Monte Carlo \\
\midrule
$\mathbb{E}[D_1]/V$ & $2/\pi=0.636620$ & $0.636830$ \\
$\mathbb{E}[D_1^2]/V^2$ & $1/2=0.500000$ & $0.500073$ \\
$\mathbb{E}[\tau_{0.03V}]$ & Bellman: $\approx9.50$ & $9.5143$ \\
$\operatorname{median}(\tau_{0.03V})$ & -- & $7$ \\
\bottomrule
\end{tabular}
\end{table}

The exact first-round moments show that maximal safe random jumps produce substantial separation immediately, even when the agents start without a defined separation direction. The near-boundary results provide a complementary finite-round characterization. Although Theorem~\ref{thm:near-boundary-hitting} gives the conservative bound
\begin{equation}
\mathbb{E}[\tau_{0.03V}]\le\frac{2}{(0.03)^2}\approx2222.22,
\end{equation}
the Bellman computation and the one-million-run Monte Carlo experiment independently estimate the same mean hitting time at approximately $9.5$ rounds. The large gap reflects the conservatism of replacing the state-dependent drift by its minimum pre-hitting value, rather than slow behavior of the actual process.

\section{Discussion and Limitations}
The policy separates safety from dispersion. The radius is a deterministic certificate; random directions explore the safe set. For general swarms, pairs are coupled through farthest-neighbor ranges. Our general guarantees are edge preservation, connectivity, boundary locking, and the boundary-edge lower bound conditional on absorption. The experiments show topology-dependent stretching but not general convergence.

Preserving every edge is deliberately conservative. Complete graphs retain every pairwise constraint and therefore cannot exceed diameter $V$, while sparse connected graphs can expand over many visibility ranges. Safely releasing redundant constraints from anonymous current ranges remains open.

We assume exact ranges and synchronous motion. Noise, actuation errors, asynchrony, collisions, and dynamics are outside scope. Existing range-only and communication-aware systems illustrate complementary mechanisms using estimation, barrier functions, online learning, and optimization \cite{chen2025rangeonly,yang2024online}. Robustification requires a noise-aware safety margin and delayed or asynchronous analysis.

\subsection{Interpretation of the General-Swarm Results}
For larger swarms, the radius of each endpoint of an edge depends on its own farthest visible neighbor, not necessarily on the edge under examination. Pairwise progress is therefore coupled through the rest of the graph, and the scalar two-agent drift does not extend directly. The experiments indicate that the initial topology strongly constrains the attainable geometry: preserving many initial edges imposes stronger geometric confinement, while sparse connected structures permit expansion over several visibility ranges.

The observed ordering should be interpreted under the stated initialization protocol. It is not a theorem that edge density alone determines the final diameter, nor does the experiment establish convergence to an absorbing state. A theoretical characterization of limiting configurations for coupled swarms remains open.

\subsection{Limitations of Preserving Every Edge}
The never-lose-a-neighbor requirement is deliberately conservative. Without bearings, identifiers, communication, or graph topology, an agent cannot certify that a particular visible edge is redundant. Consequently, redundant initial edges remain active constraints throughout execution. A complete initial graph preserves every pairwise distance below $V$ and therefore cannot exceed diameter $V$, while a sparse connected graph can expand over several visibility ranges. Designing a range-only mechanism that safely releases redundant edges remains an open problem.

\subsection{Scope and Future Robustification}
The exact range and synchronous-jump assumptions isolate the information question addressed here: what finite motion can be certified without bearings? They do not model all limitations of a physical platform. A robust implementation would need a safety margin derived from a stated range-error model, treatment of delayed or asynchronous observations, and motion execution consistent with robot dynamics. Collision avoidance would also require an additional mechanism because the present objective preserves visibility edges rather than enforcing positive pairwise clearance.

\section{Conclusion}
We derived the maximal isotropic displacement certifiable from anonymous current ranges while preserving every visibility edge under synchronous finite motion. The resulting farthest-neighbor rule uses one scalar statistic, preserves connectivity for arbitrary swarm size, and creates permanent anchors when edges reach the visibility limit. Every absorbing configuration contains at least $n/2$ boundary edges.

For two agents, maximal safe random jumps yield positive conditional drift in squared distance, almost-sure convergence to the boundary, and finite expected time to reach every fixed boundary neighborhood. Exact first-round moments, one million Monte Carlo runs, drift validation, and an independent Bellman calculation provide mutually consistent analytical and numerical characterizations. For general swarms, aggregate experiments reproduce the deterministic safety guarantee and show a consistent topology-dependent ordering of attainable diameter across the tested swarm sizes. Establishing convergence for coupled swarms, characterizing limiting configurations, and developing robust noise-aware variants remain open problems.

% =====================================================================
% Extended-version appendices. Original Sections I--VIII and Equations
% (1)--(28) are unchanged. Appendix equations use (A.1), (B.1), etc.
% =====================================================================
\newcounter{extendedappendix}
\renewcommand{\theextendedappendix}{\Alph{extendedappendix}}
\newcommand{\extendedappendix}[2]{%
  \refstepcounter{extendedappendix}%
  \setcounter{equation}{0}%
  \renewcommand{\theequation}{\theextendedappendix.\arabic{equation}}%
  \section*{Appendix \theextendedappendix: #1}%
  \label{#2}%
}

\extendedappendix{Probability Notation and Tools Used in the Paper}{app:probability-tools}

This appendix collects the probability facts used in Section~V and states them in the form needed here. It is intended to make the paper readable without requiring the reader to reconstruct intermediate probabilistic steps.

\subsection*{Conditional expectation and the filtration}
The symbol $\mathcal F_t$ denotes all information determined immediately before the round-$t$ directions are sampled. Conditional on $\mathcal F_t$, the positions $p_i(t)$, the relative displacement $X_t$, the distance $D_t$, and the step lengths are fixed. The fresh directions $U_1(t)$ and $U_2(t)$ remain random.

For a random variable $Y$, the conditional expectation $\mathbb E[Y\mid\mathcal F_t]$ is the average of $Y$ over the remaining randomness after the history through time $t$ has been fixed. If $Z$ is already determined by $\mathcal F_t$, then
\begin{equation}
\mathbb E[ZY\mid\mathcal F_t]=Z\,\mathbb E[Y\mid\mathcal F_t].
\end{equation}
This is why $D_t$, $X_t$, and $(V-D_t)/2$ can be taken outside the conditional expectations in the drift derivations.

\subsection*{Tower property}
The law of total expectation, also called the tower property, states that
\begin{equation}
\mathbb E\!\left[\mathbb E[Y\mid\mathcal F_t]\right]=\mathbb E[Y].
\end{equation}
The inner expectation averages over the next random move for each fixed past, while the outer expectation averages over all possible past histories.

\subsection*{Markov's inequality}
If $Y\ge0$ and $a>0$, then
\begin{equation}
\Pr(Y\ge a)\le\frac{\mathbb E[Y]}{a}.
\end{equation}
Indeed, $Y\ge a\mathbf 1_{\{Y\ge a\}}$, and taking expectations gives the inequality. In Theorem~4 we use $Y=(V-D_t)^2$ and $a=\varepsilon^2$.

\subsection*{First Borel--Cantelli lemma}
For events $A_0,A_1,\ldots$, if
\begin{equation}
\sum_{t=0}^{\infty}\Pr(A_t)<\infty,
\end{equation}
then, with probability one, only finitely many of the events occur. Independence is not required for this direction. Here $A_t=\{V-D_t\ge\varepsilon\}$.

\subsection*{Monotone convergence for truncated hitting times}
For a nonnegative random time $\tau$, the variables $T\wedge\tau$ increase pointwise to $\tau$. Hence
\begin{equation}
\lim_{T\to\infty}\mathbb E[T\wedge\tau]=\mathbb E[\tau].
\end{equation}
A uniform finite upper bound on the left proves that the expected hitting time is finite.

\extendedappendix{Full Derivation of the Two-Agent Drift}{app:drift-derivation}

Let
\begin{equation}
X_t=p_2(t)-p_1(t),\qquad D_t=\|X_t\|,
\end{equation}
and
\begin{equation}
r_t^{(\alpha)}=\alpha\frac{V-D_t}{2},\qquad 0\le\alpha\le1.
\end{equation}
The relative-position update is
\begin{equation}
X_{t+1}=X_t+r_t^{(\alpha)}(U_2(t)-U_1(t)).
\end{equation}
Using $\|a+b\|^2=\|a\|^2+2\langle a,b\rangle+\|b\|^2$ gives
\begin{align}
D_{t+1}^2-D_t^2
&=2r_t^{(\alpha)}\langle X_t,U_2(t)-U_1(t)\rangle\nonumber\\
&\quad+\bigl(r_t^{(\alpha)}\bigr)^2\|U_2(t)-U_1(t)\|^2.
\end{align}
Conditional on $\mathcal F_t$, $X_t$ and $r_t^{(\alpha)}$ are fixed. Since each direction is uniform on the unit circle, $\mathbb E[U_i(t)]=0$, and the conditional expectation of the linear term is zero. Moreover,
\begin{align}
\|U_2-U_1\|^2
&=2-2\langle U_1,U_2\rangle,
\end{align}
and independence and zero mean imply $\mathbb E[\langle U_1,U_2\rangle]=0$. Thus
\begin{equation}
\mathbb E[\|U_2-U_1\|^2]=2.
\end{equation}
Consequently,
\begin{align}
\mathbb E[D_{t+1}^2-D_t^2\mid\mathcal F_t]
&=2\bigl(r_t^{(\alpha)}\bigr)^2\nonumber\\
&=\frac{\alpha^2(V-D_t)^2}{2}.
\end{align}
This is Equation~(7). Setting $\alpha=1$ and adding $D_t^2$ gives Equation~(14). The value $\sqrt2$ is the root-mean-square chord length between two independent uniform points on the unit circle; it is not the mean chord length, which equals $4/\pi$.

\extendedappendix{Detailed Almost-Sure Convergence Argument}{app:convergence-proof}

Taking ordinary expectations in Equation~(14) and applying the tower property yields
\begin{equation}
\mathbb E[D_{t+1}^2]-\mathbb E[D_t^2]
=\frac12\mathbb E[(V-D_t)^2].
\end{equation}
Summing from $t=0$ to $T-1$ gives
\begin{equation}
\mathbb E[D_T^2]-D_0^2
=\frac12\sum_{t=0}^{T-1}\mathbb E[(V-D_t)^2].
\end{equation}
Since $D_T\le V$,
\begin{equation}
\sum_{t=0}^{T-1}\mathbb E[(V-D_t)^2]
\le2(V^2-D_0^2).
\end{equation}
The partial sums are nondecreasing and bounded, so letting $T\to\infty$ gives Equation~(17).

Fix $\varepsilon>0$. Markov's inequality gives
\begin{align}
\Pr(V-D_t\ge\varepsilon)
&=\Pr((V-D_t)^2\ge\varepsilon^2)\nonumber\\
&\le\frac{\mathbb E[(V-D_t)^2]}{\varepsilon^2}.
\end{align}
Summation and Equation~(17) imply
\begin{equation}
\sum_{t=0}^{\infty}\Pr(V-D_t\ge\varepsilon)<\infty.
\end{equation}
Borel--Cantelli therefore implies that $V-D_t\ge\varepsilon$ occurs only finitely often almost surely. Apply this to $\varepsilon_m=1/m$, $m\in\mathbb N$. The countable intersection of the resulting probability-one events still has probability one, and on it $D_t\to V$.

\extendedappendix{Exact Non-Hitting and Near-Boundary Hitting}{app:hitting-details}

\subsection*{Why the exact boundary is not hit in finite time}
Condition on $\mathcal F_t$ and suppose $D_t<V$. Then
\begin{equation}
D_{t+1}=\left\|X_t+\frac{V-D_t}{2}(U_2-U_1)\right\|.
\end{equation}
The relative-step magnitude is at most $V-D_t$. Reaching $V$ requires equality both in this bound and in the triangle inequality. Hence $U_1$ and $U_2$ must be antipodal, and the resulting relative step must align exactly with the current separation direction. These constraints form a lower-dimensional subset of $[0,2\pi)^2$ and have probability zero. When $X_t=0$, antipodality remains necessary and is again a measure-zero condition. A countable union over finite rounds therefore has probability zero.

\subsection*{Stopped drift for the near-boundary region}
Let
\begin{equation}
\tau_\varepsilon=\inf\{t:D_t\ge V-\varepsilon\},
\qquad
\sigma_T=T\wedge\tau_\varepsilon.
\end{equation}
For every $t<\tau_\varepsilon$, $V-D_t>\varepsilon$, and Equation~(14) gives a conditional drift of at least $\varepsilon^2/2$. Using the indicator $\mathbf 1_{\{t<\sigma_T\}}$ and summing gives
\begin{align}
\mathbb E[D_{\sigma_T}^2]-d_0^2
&=\sum_{t=0}^{T-1}\mathbb E\!\left[
\mathbf 1_{\{t<\sigma_T\}}(D_{t+1}^2-D_t^2)
\right]\nonumber\\
&\ge\frac{\varepsilon^2}{2}
\sum_{t=0}^{T-1}\Pr(t<\sigma_T).
\end{align}
For a nonnegative integer-valued random variable bounded by $T$,
\begin{equation}
\mathbb E[\sigma_T]=\sum_{t=0}^{T-1}\Pr(t<\sigma_T).
\end{equation}
Hence
\begin{equation}
\mathbb E[D_{\sigma_T}^2]-d_0^2
\ge\frac{\varepsilon^2}{2}\mathbb E[\sigma_T].
\end{equation}
Since $D_{\sigma_T}\le V$,
\begin{equation}
\mathbb E[T\wedge\tau_\varepsilon]
\le\frac{2(V^2-d_0^2)}{\varepsilon^2}.
\end{equation}
Finally, monotone convergence gives Equation~(19).

\extendedappendix{Random-Chord Moments for the First Round}{app:random-chord}

Suppose $D_0=0$. Each agent moves $V/2$ from the common starting point. If their angular difference is $\Delta$, the chord formula gives
\begin{equation}
D_1=V\left|\sin\left(\frac{\Delta}{2}\right)\right|.
\end{equation}
The relative angle modulo $2\pi$ is uniform on $[0,2\pi)$. Equivalently, the smaller angle $\Phi\in[0,\pi]$ is uniform on $[0,\pi]$, so $\mathbb E[\Phi]=\pi/2$. This does not imply that the mean chord equals the chord at $\pi/2$: since $D_1/V=\sin(\Phi/2)$,
\[
\mathbb E\!\left[\frac{D_1}{V}\right]
=\mathbb E\!\left[\sin\!\left(\frac{\Phi}{2}\right)\right]
\neq
\sin\!\left(\frac{\mathbb E[\Phi]}{2}\right)
=\frac{1}{\sqrt2}.
\]
The left-hand side is smaller because $\phi\mapsto\sin(\phi/2)$ is strictly concave on $(0,\pi)$. Therefore
\begin{align}
\mathbb E[D_1]
&=\frac{V}{2\pi}\int_0^{2\pi}
\left|\sin\left(\frac{\delta}{2}\right)\right|d\delta
=\frac{2V}{\pi},
\end{align}
and
\begin{align}
\mathbb E[D_1^2]
&=\frac{V^2}{2\pi}\int_0^{2\pi}
\sin^2\left(\frac{\delta}{2}\right)d\delta
=\frac{V^2}{2}.
\end{align}
Thus the mean distance is $2V/\pi$, whereas the root-mean-square distance is $V/\sqrt2$.

\extendedappendix{Experimental Metrics and the Fixed Initial Spanning Tree}{app:experimental-metrics}

The fixed initial spanning tree is a measurement device and is not part of the motion rule. At time zero, one spanning tree $T_0$ is selected from the initial visibility graph. It contains all $N$ agents and exactly $N-1$ initial edges. The endpoint identities are stored and remain fixed; the tree is never recomputed, even if new visibility edges appear later.

For a stored edge $(i,j)\in T_0$, define
\begin{equation}
d_{ij}(t)=\|p_i(t)-p_j(t)\|,
\end{equation}
and the relative stretching ratio
\begin{equation}
R_{ij}(t)=\frac{d_{ij}(t)}{d_{ij}(0)}.
\end{equation}
Values above one denote stretching and values below one denote contraction. Since every topology contributes exactly $N-1$ tracked edges, the metric avoids comparing the complete graph's many edges with the path's minimal edge set.

The aggregate statistic reported in the general-swarm experiments is based on the mean length of the fixed tree,
\begin{equation}
\overline d_{T_0}(t)=\frac{1}{N-1}\sum_{(i,j)\in T_0}d_{ij}(t),
\end{equation}
and, when a relative stretching value is shown, on
\begin{equation}
S_{T_0}(t)=\frac{\overline d_{T_0}(t)}{\overline d_{T_0}(0)}.
\end{equation}
This definition first averages the same $N-1$ tracked edge lengths and then forms the final-to-initial ratio. It should not be confused with averaging the individual ratios $R_{ij}(t)$, which can give disproportionate weight to initially short edges.

This metric complements geometric diameter,
\begin{equation}
\operatorname{diam}(t)=\max_{i,j}\|p_i(t)-p_j(t)\|.
\end{equation}
Diameter measures the largest global separation; fixed-tree stretching measures deformation of one initially connectivity-sufficient local structure. A chain can unfold and greatly increase diameter even when some tracked edges contract, while tracked edges can stretch without producing a proportional diameter increase. Ratios should also be interpreted with normalized absolute lengths because an initially very short edge can have a large ratio while remaining far below $V$.

\clearpage
\onecolumn
\extendedappendix{Geometric Simulation Illustrations}{app:geometric-illustrations}

The aggregate experiments quantify behavior across independent runs. The representative executions below provide complementary geometric intuition and are not used as evidence for convergence.

\begin{figure}[!ht]
\centering
\includegraphics[width=0.92\textwidth]{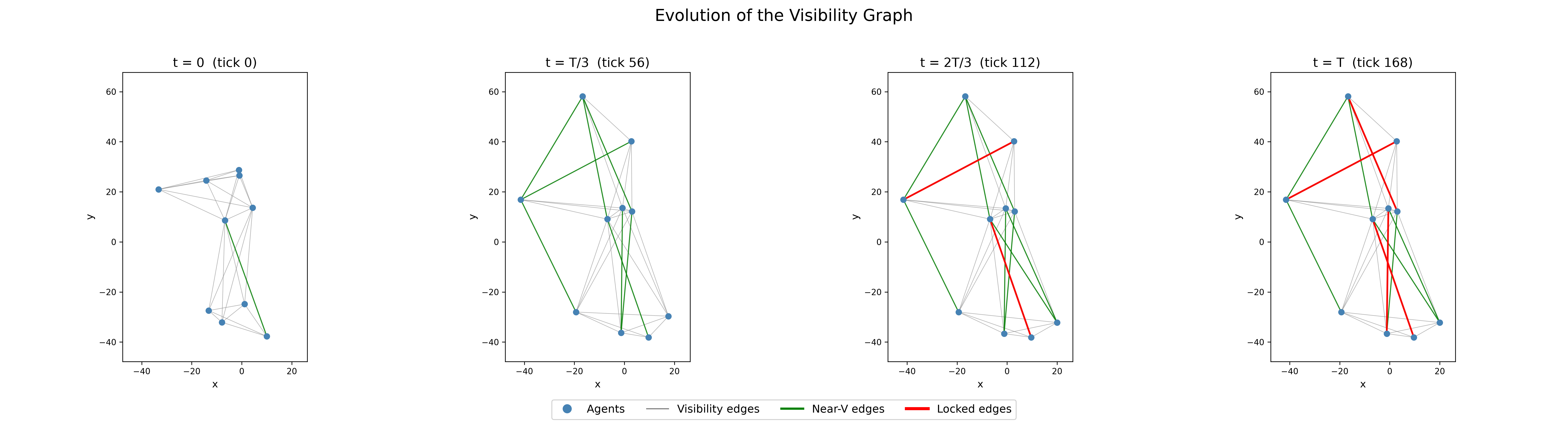}
\caption{Four snapshots from a representative random-connected execution. Existing visibility edges are preserved while the configuration expands and develops longer edges. This single execution illustrates the geometry and does not replace the aggregate topology study.}
\label{fig:representative-evolution-extended}
\end{figure}

\begin{figure}[!ht]
\centering
\includegraphics[width=0.84\textwidth]{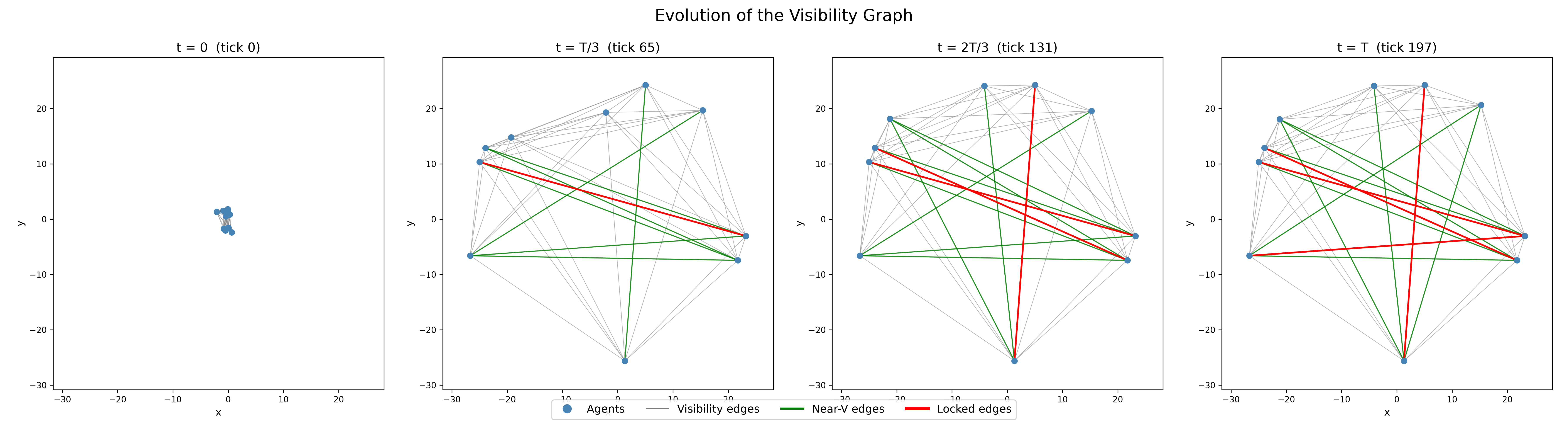}
\vspace{-1mm}

{\footnotesize (a) Complete initial visibility graph.}
\vspace{2mm}

\includegraphics[width=0.84\textwidth]{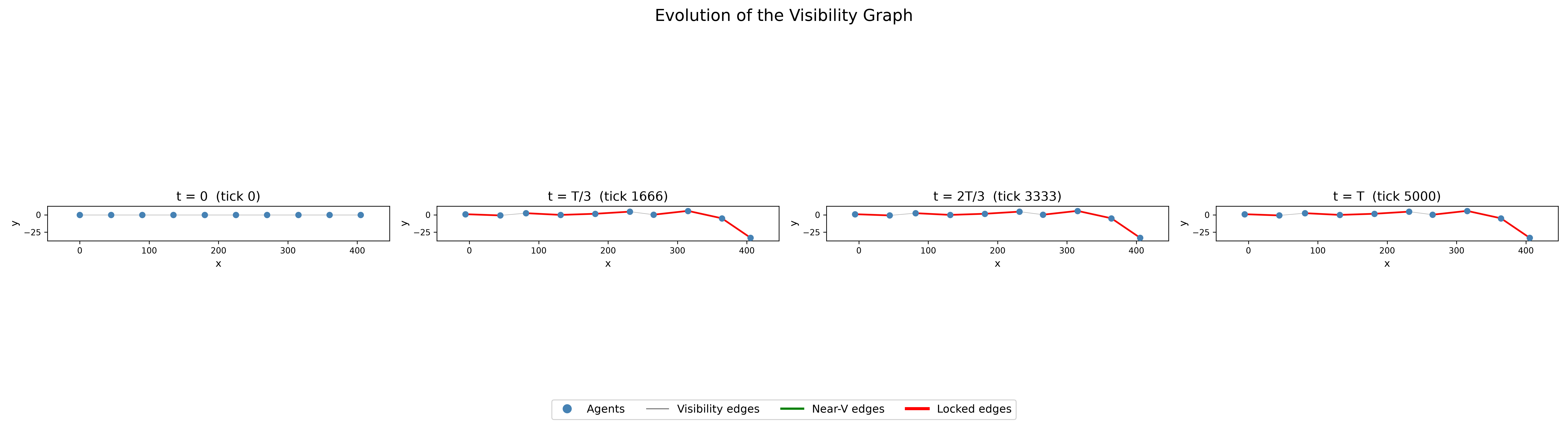}
\vspace{-1mm}

{\footnotesize (b) Sparse path-like initial visibility graph.}
\caption{Geometric effect of the initial topology. Preserving every edge of a complete graph confines the diameter to at most $V$. A sparse path-like graph preserves only local consecutive constraints and can expand over several visibility ranges; the illustrated execution reaches approximately $8.25V$.}
\label{fig:complete-versus-sparse-extended}
\end{figure}

\begin{table}[!ht]
\centering
\caption{Interpretation of the representative topology extremes.}
\label{tab:topology-interpretation}
\begin{tabular}{p{0.18\textwidth}p{0.34\textwidth}p{0.37\textwidth}}
\toprule
Initial topology & Permanently preserved constraints & Geometric consequence \\
\midrule
Complete & Every initial pair remains at distance at most $V$. & The geometric diameter remains at most $V$, independently of swarm size. \\
Path-like & Only consecutive initial pairs are directly constrained by the preserved path edges. & The end-to-end distance may span several visibility ranges, although planar folding can reduce the attained diameter. \\
\bottomrule
\end{tabular}
\end{table}

\clearpage
\twocolumn
\bibliographystyle{IEEEtran}
\bibliography{references}

@book{durrett2019probability,
  author    = {Rick Durrett},
  title     = {Probability: Theory and Examples},
  edition   = {5},
  series    = {Cambridge Series in Statistical and Probabilistic Mathematics},
  publisher = {Cambridge University Press},
  year      = {2019},
  isbn      = {9781108473682}
}

@article{ando1999distributed,
  author  = {Hideki Ando and Yoshinobu Oasa and Ichiro Suzuki and Masafumi Yamashita},
  title   = {Distributed Memoryless Point Convergence Algorithm for Mobile Robots with Limited Visibility},
  journal = {IEEE Transactions on Robotics and Automation},
  volume  = {15},
  number  = {5},
  pages   = {818--828},
  year    = {1999},
  doi     = {10.1109/70.795787}
}

@article{ji2007distributed,
  author  = {Meng Ji and Magnus Egerstedt},
  title   = {Distributed Coordination Control of Multiagent Systems While Preserving Connectedness},
  journal = {IEEE Transactions on Robotics},
  volume  = {23},
  number  = {4},
  pages   = {693--703},
  year    = {2007},
  doi     = {10.1109/TRO.2007.900638}
}

@inproceedings{gordon2004gathering,
  author    = {Noam Gordon and Israel A. Wagner and Alfred M. Bruckstein},
  title     = {Gathering Multiple Robotic A(ge)nts with Limited Sensing Capabilities},
  booktitle = {Ant Colony Optimization and Swarm Intelligence},
  series    = {Lecture Notes in Computer Science},
  volume    = {3172},
  pages     = {142--153},
  publisher = {Springer},
  year      = {2004},
  doi       = {10.1007/978-3-540-28646-2_13}
}

@inproceedings{gordon2008crude,
  author    = {Noam Gordon and Yotam Elor and Alfred M. Bruckstein},
  title     = {Gathering Multiple Robotic Agents with Crude Distance Sensing Capabilities},
  booktitle = {Ant Colony Optimization and Swarm Intelligence},
  series    = {Lecture Notes in Computer Science},
  volume    = {5217},
  pages     = {72--83},
  publisher = {Springer},
  year      = {2008},
  doi       = {10.1007/978-3-540-87527-7_7}
}

@article{bellaiche2017continuous,
  author  = {Levi-Itzhak Bellaiche and Alfred M. Bruckstein},
  title   = {Continuous Time Gathering of Agents with Limited Visibility and Bearing-Only Sensing},
  journal = {Swarm Intelligence},
  volume  = {11},
  number  = {3--4},
  pages   = {271--293},
  year    = {2017},
  doi     = {10.1007/s11721-017-0140-y}
}

@misc{manor2019local,
  author       = {Manor, Rotem and Barel, Ariel and Bruckstein, Alfred M.},
  title        = {Local Interactions for Cohesive Flexible Swarms},
  year         = {2019},
  howpublished = {arXiv preprint arXiv:1903.09259},
  note         = {Available: https://arxiv.org/abs/1903.09259},
  doi          = {10.48550/arXiv.1903.09259}
}

@article{zavlanos2007potential,
  author={Zavlanos, Michael M. and Pappas, George J.},
  title={Potential Fields for Maintaining Connectivity of Mobile Networks},
  journal={IEEE Transactions on Robotics}, volume={23}, number={4}, pages={812--816}, year={2007},
  doi={10.1109/TRO.2007.900642}
}

@inproceedings{dimarogonas2008decentralized,
  author={Dimarogonas, Dimos V. and Johansson, Karl H.},
  title={Decentralized Connectivity Maintenance in Mobile Networks with Bounded Inputs},
  booktitle={Proceedings of the 2008 IEEE International Conference on Robotics and Automation},
  pages={1507--1512}, year={2008}, doi={10.1109/ROBOT.2008.4543415}
}

@article{flocchini2005gathering,
  author={Flocchini, Paola and Prencipe, Giuseppe and Santoro, Nicola and Widmayer, Peter},
  title={Gathering of Asynchronous Robots with Limited Visibility},
  journal={Theoretical Computer Science}, volume={337}, pages={147--168}, year={2005},
  doi={10.1016/j.tcs.2005.01.001}
}

@article{lee2023practical,
  author={Lee, Chan-Seok and Suh, Ui-Suk and Lee, Kang-Min and Whang, Ick-Ho and Ra, Won-Sang},
  title={Practical Distance-Based Multi-robot Formation Control Using Low-Cost Ultrasonic Source Tracker},
  journal={Journal of Electrical Engineering \& Technology}, volume={18}, number={4}, pages={3219--3236}, year={2023},
  doi={10.1007/s42835-023-01431-0}
}

@inproceedings{li2022decentralized,
  author={Li, Minghao and Jie, Yingrui and Kong, Yang and Cheng, Hui},
  title={Decentralized Global Connectivity Maintenance for Multi-Robot Navigation: A Reinforcement Learning Approach},
  booktitle={Proceedings of the 2022 International Conference on Robotics and Automation},
  pages={8801--8807}, year={2022}, doi={10.1109/ICRA46639.2022.9812163}
}

@inproceedings{yang2024online,
  author={Yang, Yupeng and Lyu, Yiwei and Zhang, Yanze and Gao, Ian and Luo, Wenhao},
  title={Integrating Online Learning and Connectivity Maintenance for Communication-Aware Multi-Robot Coordination},
  booktitle={Proceedings of the 2024 IEEE/RSJ International Conference on Intelligent Robots and Systems},
  pages={5770--5776}, year={2024}, doi={10.1109/IROS58592.2024.10802189}
}

@article{chen2025rangeonly,
  author  = {Chen, Jin and Li, Ming and Marcantoni, Matteo and
             Jayawardhana, Bayu and Wang, Yafei},
  title   = {Range-Only Distributed Safety-Critical Formation Control
             Based on Contracting Bearing Estimators and
             Control Barrier Functions},
  journal = {IEEE Internet of Things Journal},
  volume  = {12},
  number  = {19},
  pages   = {40968--40979},
  year    = {2025},
  doi     = {10.1109/JIOT.2025.3590774}
}

@inproceedings{brandstatter2024flock,
  author={Brandst\"atter, Andreas and Smolka, Scott A. and Stoller, Scott D. and Tiwari, Ashish and Grosu, Radu},
  title={Flock-Formation Control of Multi-Agent Systems Using Imperfect Relative Distance Measurements},
  booktitle={Proceedings of the 2024 IEEE International Conference on Robotics and Automation},
  pages={12193--12200}, year={2024}, doi={10.1109/ICRA57147.2024.10610147}
}

@misc{barel2019cometogether,
  author       = {Barel, Ariel and Manor, Rotem and Bruckstein, Alfred M.},
  title        = {{COME TOGETHER}: Multi-Agent Geometric Consensus
                  (Gathering, Rendezvous, Clustering, Aggregation)},
  year         = {2019},
  howpublished = {arXiv preprint arXiv:1902.01455},
  note         = {Available: https://arxiv.org/abs/1902.01455}
}
\end{document}